\documentclass[10pt,twoside]{article}

\usepackage[ansinew]{inputenc} 
\usepackage{amsmath,amsfonts,amssymb,amsthm}
\usepackage{bbm}

\newcommand{\BIGOP}[1]{\mathop{\mathchoice%
{\raise-0.22em\hbox{\huge $#1$}}%
{\raise-0.05em\hbox{\Large $#1$}}{\hbox{\large $#1$}}{#1}}}
\newcommand{\bigtimes}{\BIGOP{\times}}
\newcommand{\BIGboxplus}{\mathop{\mathchoice%
{\raise-0.35em\hbox{\huge $\boxplus$}}%
{\raise-0.15em\hbox{\Large $\boxplus$}}{\hbox{\large $\boxplus$}}{\boxplus}}}

\newcommand{\field}[1]{\mathbb{#1}}
\newcommand{\N}{\field{N}}
\newcommand{\R}{\field{R}}
\newcommand{\C}{\field{C}}
\newcommand{\Z}{\field{Z}}

\newcommand{\HH}{\mathcal H}

\newcommand{\EE}{\mathcal E}
\newcommand{\FF}{\mathcal F}

\newcommand{\QQ}{\mathcal Q}

\newcommand{\eps}{\varepsilon}
\newcommand{\ph}{\varphi}

\newcommand{\sprod}[2]{\mbox{$\left\langle #1,#2 \right\rangle$}}        
\newcommand{\form}[3]{\mbox{$\left\langle #1\left|#2 \right|#3 \right\rangle $}}   

\newcommand{\supp}{\operatorname{supp}}

\newcommand{\Rea}{\operatorname{Re}}

\newcommand{\curl}{\operatorname{curl}}

\newcommand{\dist}{\operatorname{dist}}

\newtheorem{theorem}{Theorem}[section]
\newtheorem{lemma}[theorem]{Lemma}
\newtheorem{corollary}[theorem]{Corollary}

\newtheorem{prop}[theorem]{Proposition}
\theoremstyle{plain}

\title{On the Strong Coupling Limit of Many-Polaron Systems in Electromagnetic Fields}

\author{
D.~Wellig\\
Universit\"at Stuttgart, Fachbereich Mathematik\\
70550 Stuttgart, Germany}

\date{\today}
\begin{document}

\maketitle

\begin{abstract}
In this paper estimates on the ground state energy of Fr\"ohlich $N$-polarons in electromagnetic fields in the strong coupling limit, $\alpha\to\infty$, are derived. It is shown that the ground state energy is given by $\alpha^2$ multiplied by the minimal energy of the corresponding Pekar-Tomasevich functional for $N$ particles, up to an error term of order $\alpha^{42/23}N^3$. The potentials $A,V$ are suitably rescaled in $\alpha$. As a corollary, binding of $N$-polarons for strong magnetic fields for large coupling constants is established.
\end{abstract}

\section{Introduction and Main Results}


An ionic crystal is deformed by the presence of an excess electron via the Coulomb attraction resp. repulsion. The distortion induces a potential which acts on the electron. The resulting composite particle is called a polaron. More generally a $N$-polaron is a system of $N$ electrons with the corresponding distortions of the ionic lattice. In the physically admissible region the coupling constant $\alpha$ between electron and lattice, in our units, is bounded from above by the electron-electron repulsion strenght $U$. Energetically it is more favorable if the electrons deform the lattice in a small region, hence they tend to stay close together. Therefore an attractive force operates between the electrons which is counteracted by their Coulomb repulsion. Which force is stronger, depending on $\alpha$ and $U$, is discussed further below. For more information about the physical properties of polarons we refer to \cite{Dev1996, AD2010} and references therein.

The goal of this work is to prove that in the leading order of the coupling constant the ground state energy of $N$-polarons subject to a certain class of electromagnetic fields is given by the minimal energy of the Pekar-Tomasevich functional. For large values of $\alpha$ the effect of the external fields is negligible. Hence they are rescaled such that they grow with increasing $\alpha$. Combining this with the binding of Pekar-Tomasevich $N$-polarons subject to a constant magnetic field, which was recently established in \cite{AG2013}, we prove binding for Fr\"ohlich $N$-polarons in strong constant magnetic fields for large couplings. In the $N$-particle case without external fields, similar asymptotic exactness and binding results have recently been derived in \cite{AL2012}. The common strategy of the latter work and ours is to split up the $N$-polaron into disjoint groups of polarons, to estimate the interaction energy between the groups, and to derive the asymptotic coincidence with Pekar-Tomasevich for the individual groups by the techniques developed by Lieb and Thomas \cite{LT1997}.

We consider the model, introduced by H. Fr\"ohlich \cite{F1954}, that describes large polarons, i.e. polarons with large spatial extension compared to the lattice spacing. Additionally external potentials $V:\R^3\to\R$ and $A:\R^3\to\R^3$ are introduced, which generate the electric field $-\nabla V$ and the magnetic field $\curl A$. 
The Fr\"ohlich hamilton operator for $N$-polarons on the Hilbert space $\HH = L^2(\R^{3N})\otimes\FF$, with $\FF$ as the bosonic Fock space over $L^2(\R^3)$, is given by
\begin{align}
H^{(N)} &= \sum_{j=1}^N \left(D_{A,x_j}^2+V(x_j)+\sqrt\alpha\phi (x_j)\right)+H_{ph}+ UV_C(x_1,\ldots ,x_N), \label{frhamiltonian}
\end{align}
where $D_{A,x_j} = -i\nabla_{x_j}+A(x_j)$. $V_C(x_1,\ldots ,x_N)= \sum_{i<j} \frac{1}{|x_i-x_j|}$ is the Coulomb potential and the interaction between the electron and the quantized lattice vibrations (i.e. phonons) is
$$
\phi (x) = \frac{1}{\sqrt 2 \pi}\int\frac{dk}{|k|}\left(a(k)e^{ikx}+a^{\ast}(k)e^{-ikx}\right).
$$
Where $a(k)$ represents the creation- and $a^{\ast}(k)$ the annihilation operator with momentum $k$ and $H_{ph}=\int_{\R^3}dk a^{\ast}(k)a(k)$ denotes the phonon energy. The ground state energy of $H^{(N)}$ is defined by
$$
E^{(N)} (A,V,U,\alpha) = \inf_{\|\psi\|=1}\form{\psi}{H^{(N)}}{\psi}.
$$
For simplicity reasons $E^{(N)} (A,V,U,\alpha)$ sometimes is written as $E^{(N)}$. Because of Lemma~\ref{lm:cutoff} $E^{(N)} (A,V,U,\alpha)$ is bounded from below and hence it exists.

The Fr\"ohlich model is closely related to the Pekar-Tomasevich functional $\EE^{(N)}_{U,\alpha}(A,V,.)$, which for normalized $\ph\in L^2(\R^{3N})$ may be defined by
\begin{align}
\EE^{(N)}_{U,\alpha}(A,V,\ph) = \inf_{\|\eta\|=1}\form{\ph\otimes\eta}{H^{(N)}}{\ph\otimes\eta}.\label{PT-derivation}
\end{align}
See Section~\ref{sec:compsup} for a more explicit definition. The minimal energy of the Pekar-Tomasevich functional with external magnetic and electric fields is denoted by
$$
C_N(A,V,U,\alpha) = \inf_{\|\ph\|=1} \EE^{(N)}_{U,\alpha}(A,V,\ph).
$$ 
Sometimes the short hand $C_N$ instead of $C_N(A,V,U,\alpha)$ is used. By \eqref{PT-derivation} 
\begin{align}
C_N\ge E^{(N)}. \label{ineq:estabove}
\end{align}
The dimensionless constant $\nu:= U/\alpha$ describes the physical region for $\nu >2$. Does there exist a similar estimate converse to \eqref{ineq:estabove}? In the following theorem we give an affirmative answer.

\begin{theorem}\label{mainthm}
For any values of $\nu>0$ and $N$, the following is true:
\begin{itemize}
\item[(a)] Suppose $A, V$ satisfy assumptions (AV1) and \eqref{generalenergyin} described in Section~\ref{sec:proofthm}, then there exists $c(A,V)$
$$
E^{(N)}(A_{\alpha},V_{\alpha},\alpha\nu,\alpha) \ge  \alpha^2 C_N(A,V,\nu,1) - c(A,V) \alpha^{42/23}N^3,
$$
for $\alpha$ large and $A_{\alpha}(x) = \alpha A (\alpha x)$, $V_{\alpha}(x) = \alpha^2 V(\alpha x)$.
\item[(b)] Suppose $A, V$ satisfy assumptions (AV2) and \eqref{generalenergyin} described in Section~\ref{sec:proofthm}, then
\begin{align}
\lim_{\alpha\to\infty} \alpha^{-2} E^{(N)}(A,V,\alpha\nu,\alpha) = C_N(0,0,\nu,1),\quad (\alpha\to\infty),\quad \text{for all } N.\label{eqn:thmb}
\end{align}

\end{itemize}
\end{theorem} 


Theorem~\ref{mainthm}~b) shows what one physically would expect, that the ground state energy of the Fr\"ohlich model does not depend on the external fields in the leading order of the coupling constant for $\alpha\to\infty$. The external fields are rescaled such that they are appreciable for large $\alpha$. Furthermore the scaling property from Theorem~\ref{mainthm}~a) ensures that the minimal energy of the electromagnetic Pekar functional is proportional to $\alpha^2$, i.e.
\begin{align}
C_N(A_{\alpha} ,V_{\alpha} ,\alpha\nu ,\alpha) = \alpha^2 C_N(A,V,\nu,1).\label{scalingproperty}
\end{align}

In the case $N=1$ Theorem~\ref{mainthm} was recently proved in \cite{GW2013}. The previous results without external fields are discussed further below.

Furthermore, we want to study the formation of multipolarons in constant magnetic fields. Binding for $N$-polarons is established if
\begin{align}
\Delta E^{(N)} := \min_{1\le k\le N-1}\left(E^{(k)}+E^{(N-k)}\right)-E^{(N)}>0, \label{bindingcond}
\end{align}
and analogously for Pekar-Tomasevich $N$-polarons
\begin{align}
\min_{1\le k\le N-1}\left(C_k+C_{N-k}\right)-C_N>0. \label{bindingcondpek}
\end{align}
We already know binding for Pekar-Tomasevich $N$-polarons in constant magnetic fields for $\nu$ in some neighborhood of $\nu =2$ \cite{AG2013}, hence as a corollary of Theorem~\ref{mainthm}, it follows the existence of bound states for Fr\"ohlich $N$-polarons in strong constant magnetic fields for $\alpha$ large enough. Thus:

\begin{theorem}\label{mainthm2}
For any values of $N$, let $A$ be linear, i.e. the corresponding magnetic field is constant, then there exists $\nu_{N,A}>2$ such that for $\nu<\nu_{N,A}$ and $\alpha$ large enough 
$$
\Delta E^{(N)}(A_{\alpha},0,\alpha\nu,\alpha) > 0, 
$$
where $A_{\alpha}(x) = \alpha A(\alpha x)$.
\end{theorem}

\noindent\emph{Remark.} Suppose $A, V$ satisfy \eqref{generalenergyin} and assumption (AV2) described in Section~\ref{sec:proofthm}, then there exists $\nu_N>2$ such that for $\nu <\nu_N$ and $\alpha$ large enough 
\begin{align}
\Delta E^{(N)}(A,V,\alpha\nu,\alpha) > 0. \label{bindingwithoutscaling}
\end{align}
This follows from \eqref{eqn:thmb} and the binding of $N$-polarons in the Pekar-Tomasevich model without external fields (see \cite{L2011} or \cite{AG2013} for $A=0$). In other words, in the leading order for $\alpha\to\infty$ the binding energy does not depend on the (non-scaled) external fields.

For $N=1$ without external fields a first proof of Theorem~\ref{mainthm} was given by means of stochastic integration by Donsker and Varadhan \cite{DV1983}, however they did not mention an explicit error bound. Later, Lieb and Thomas \cite{LT1997} gave another proof using operator theoretical methods, which was the basis for subsequent generalizations, i.e. to the case of polarons subject to electromagnetic fields \cite{GW2013} and to the case of $N$-polarons without external fields \cite{AL2012}. Since the idea of the proof in \cite{LT1997} only applies to multipolarons in a neighborhood of one another, Theorem~\ref{mainthm} can not simply be adapted. Considering that, in Proposition~\ref{prop:split} we estimate the interaction-energy between different clusters of multipolarons by a generalization of a lemma recently appeared in \cite{FLST2011a}.

The basic idea of our proof of Theorem~\ref{mainthm2} goes back to Miyao and Spohn \cite{MS2007}, where they proved formation of bipolarons. They argued that in the strong coupling regime, binding for bipolarons is implied by the binding for Pekar-Tomasevich bipolarons and the fact that in the leading order for $\alpha\to\infty$ the Fr\"ohlich ground state energy is exactly described by the Pekar minimal energy. By a similar reasoning, the existence of bipolarons subject to electromagnetic fields was recently derived in \cite{GW2013} with the help of binding of the corresponding Pekar bipolarons \cite{GHW2012}. For binding of $N$-polarons, but without external fields we refer to \cite{AL2012, L2011}. There are further binding and non-binding results in the mathematical and physical literature. Namely non-binding for $N$-polarons without external fields have been proved for the Fr\"ohlich model and the Pekar functional for sufficiently large values of $\nu>0$ \cite{FLST2011a}. Numerical calculations suggest that binding for bipolarons does not occur for small couplings \cite{VPD90,VSPD1992}, but there exists no rigorous proof yet.\\


\noindent\emph{Acknowledgements.} It is a pleasure to thank Marcel Griesemer and Ioannis Anapolitanos for many fruitful discussions, for their comments and corrections. This work is supported by the German Science Foundation (DFG), grant GR 3213/1-1.

\section{Preparations and Structure of the Proof} \label{sec:proofthm}
Let the aforementioned external potential $V:\R^3\to\R$ be form-bounded with bound zero, i.e. for all $\eps>0$ there exists $C_{\eps}>0$
\begin{align}
|\sprod{\ph}{V\ph}|\le \eps \|\nabla\ph\|^2 + C_{\eps}\|\ph\|^2,\quad \ph\in C_0^{\infty}(\R^3).\label{Vformbound}
\end{align}
Two different assumptions on the magnetic and electric fields are given
\begin{itemize}
\item[(AV1)] $A_k\in L^2_{loc}(\R^3)$, $V\in L^1_{loc}(\R^3)$ and \eqref{Vformbound},
\item[(AV2)] $A_k\in L^3_{loc}(\R^3)$, $V\in L^{3/2}_{loc}(\R^3)$ and \eqref{Vformbound}.
\end{itemize}
Obviously (AV2) is contained in (AV1). If nothing is mentioned, always (AV1) is supposed. (AV1) ensures that the quadratic form $\sum_{j=1}^N\sprod{D_{A,x_j}\ph}{D_{A,x_j}\ph}$ on $C_0^{\infty}(\R^{3N})$ is well defined. It is closable and the domain of the closure is $H^1_A(\R^{3N}):=\{\ph\in L^2(\R^{3N})|(-i\partial_{x_j,\ell}+A_{\ell}(x_j))\ph\in L^2(\R^{3N}), 1\le \ell\le 3, 1\le j\le N\}$ (see \cite{AHS1978}).

We recall the important diamagnetic inequality that is frequently used in the present paper. It states, that if $\ph \in H^1_A(\R^3)$, then $|\ph |\in H^1(\R^3)$ and
$$
|\nabla|\ph| (x)|\le |D_A\ph (x)|,\quad \text{pointwise for almost every } x\in \R^3.
$$
For a proof see \cite{LL1997}.

(AV1) and the diamagnetic inequality imply that $\sum_{j=1}^N\sprod{D_{A,x_j}\ph}{D_{A,x_j}\ph} + \sprod{\ph}{V(x_j)\ph}$ is a closed quadratic form on $H^1_A(\R^{3N})$. Since $\int dk |k|^{-1} a^{\ast}(k)$ makes no sense on $\FF_0 =\{(\ph^{(n)})\in\FF|\ph^{(n)}\in C_0(\R^{3N}), \ph^{(n)} =0 \text{ for all but finitely many } n\}$, the Fr\"ohlich hamiltonian is not well-defined on $Q=C_0^{\infty}(\R^{3N})\otimes \FF_0$. This drawback is avoided by interpreting $H^{(N)}$ as a quadratic form $\form{\psi}{H^{(N)}}{\psi}$ on $\QQ$. Because of Lemma~\ref{lm:cutoff} it is closable and semibounded on $\QQ$, therefore the closure is a quadratic form of a self-adjoint operator.

Further, we assume the following energy inequality
\begin{align}
C_{n}+C_m \ge C_{m+n},\quad\text{for } m+n\le N.\label{generalenergyin}
\end{align}
The following choices of potentials $A,V$ satisfy \eqref{generalenergyin}.
\begin{itemize}
\item[1)] There exists $w\in\R^3$ and $f\in H^2(\R^3)$, $f(x+w)=f(x)$: $A(x+w) = A(x) +\nabla f(x)$ and $V(x+w) = V(x)$ (periodic electric potential and periodic magnetic field).\\
\textit{Proof.} Let $\ph_i\in L^2(\R^{3n_i})$, $1\le i\le 2$ be approximative minimizers of $\EE^{(n_i)}_{\nu,1}(A,V,.)$ up to an error of $\eps$. We define the discrete magnetic translation by
$$
\ph_2^k(x) = \ph_2(x+kw) e^{if(x)k},\quad k\in\Z
$$ 
then $\|D_A\ph_2^k\| = \|D_A\ph_2\|$ for all $k\in\Z$. Hence
$$
\EE^{(n_1+n_2)}_{\nu,1}(A,V,\ph_1\otimes\ph_2^k)<\sum_{i=1}^2 C_{n_i}(A,V,\nu,1) +2\eps +o(1)_{k\to\infty},
$$
where $o(1)_{k\to\infty}$ stems from the mixing terms of the self-interaction and the Coulomb interaction of the first $n_1$ and the last $n_2$ particles.
\item[2)] $A$ linear, and $V\in L^{\infty} (\R^3)$, $V\ge 0$, $\lim_{|x|\to\infty} V(x) =0$.
\end{itemize}

\textbf{Structure of the Proof.} Our proof of Theorem~\ref{mainthm} is a generalization of the $N$-polaron case without external fields \cite{AL2012}. In \cite{AL2012} the polarons are divided into clusters in order to distinguish the ones that are in a neighborhood of each other to the ones that are not. With the help of a formula from Feynman and Kac (see Lemma~1 of \cite{FLST2011a}), derived by stochastic integration, the energy of the inter-cluster interactions is bounded from above. The formula from Feynman and Kac seems not to be easily generalizable to magnetic Schr\"odinger operators. Instead, in this paper the polarons are grouped into disjoint balls with sufficiently large distances to each other and bounded radii. The localization and the regrouping into suitable balls is done in Lemma~\ref{lm:localization}. In Proposition~\ref{prop:split} then the energy of the inter-ball interactions are estimated by a generalization of Lemma~3 of \cite{FLST2011a}. 

In the next step, the energies of $N$-polarons localized in balls are bounded from below by the ground state energy of the $N$-particle Pekar-Tomasevich functional. A proof is done in Proposition~\ref{prop:liebthomas}, which is based on \cite{LT1997}. From the proof it is also clear that the estimate does not depend on the concrete centers of the balls, although the fields $A,V$ do not have to be translation invariant.

\section{Estimate of the Multipolaron Interaction Energy}\label{sec:estintenerg}

The $N$-polarons are first localized into $N$ arbitrarily distributed equal sized balls. These balls then can be grouped in the following manner: There exist bigger disjoint balls that contain the smaller ones, additionally each radius is bounded in terms of the number of smaller balls in the corresponding bigger one. The following lemma addresses this issue.

\begin{lemma}\label{lm:localization}
Suppose $R >0$, then for every normalized $\psi\in \QQ$ there exists a normalized $\psi_0\in\QQ$ satisfying 
\begin{align}
\form{\psi}{H^{(N)}}{ \psi}\ge \form{\psi_0}{H^{(N)} }{\psi_0}-\frac{9N\pi^2}{4R^2}\label{ineq:localization}
\end{align}
and $\supp\psi_0\subset \bigtimes_{i=1}^m B_i^{n_i}$, $B_i^{n_i} = \bigtimes_{j=1}^{n_i} B_i$. Here $B_i$ are balls with radius $R_i$, $n_i>0$, $\sum_{i=1}^m n_i =N$ such that
\begin{itemize}
\item[(i)] $\dist (B_i,B_j)\ge R$ for $i\ne j$,
\item[(ii)] $R_i= \frac{1}{2}(3n_i-1)R$.
\end{itemize}

\end{lemma}

\begin{proof}
In Step~1 $\psi\in \QQ$ is localized. More explicitly: We show that for every $\psi\in \QQ$ there exists $\tilde\psi_0\in\QQ$ satisfying \eqref{ineq:localization} and $\supp\tilde\psi_0\subset\bigtimes_{k=1}^N B_R(y_k)$, where $B_R(y_k)$ are balls with radius $R$ and centers $y_k\in\R^3$. Then in Step~2 we regroup the $N$ balls $B_R(y_k)$ found in Step~1 and inscribe them into disjoint bigger balls, i.e. we prove the existence of balls $B_i$, $1\le i\le m$, where $m\le N$, satisfying (i), (ii) and of a permutation $\sigma\in S_{N}$ such that $\bigtimes_{k=1}^N B_R(y_{\sigma (k)})\subset \bigtimes_{i=1}^m B_i^{n_i}$. The lemma then follows by $\psi_0(x_1,\ldots,x_N) := \tilde\psi_0( x_{\sigma^{-1}(1)},\ldots,x_{\sigma^{-1}(N)})$ and the fact that $\form{.}{H^{(N)}}{.}$ is invariant under permutations of the variables $x_1,\ldots , x_N$.\\

\noindent\emph{Proof of Step~1.} For arbitrary $L>0$, which later will be chosen as $L = 2R/\sqrt 3$, a suitable localization function on $\R^{3N}$ is defined by
$$
\phi (x) := \prod_{j=1}^{3N} \cos (x_j\pi /L)  \chi_{[-L/2, L/2]} (x_j)\quad\text{and}\quad \phi_y(x):= \phi(x-y),\quad y\in\R^{3N}.
$$
Thus $\phi_y$ is supported in a $3N$-dimensional cube of sidelength $L$ and center $y$. 

By straightforward calculations
\begin{align}
&\int_{\R^{3N}} dy \form{\phi_y\psi}{H^{(N)}}{\phi_y\psi} \nonumber\\
&= \int_{\R^{3N}} dy\left[\form{\psi}{H^{(N)}}{\psi}\|\phi_y\psi\|^2+2\Rea \sprod{(-i\nabla\phi_y)\psi}{\phi_yD_A\psi} +\|(-i\nabla\phi_y)\psi\|^2\right]\nonumber\\
&=\int_{\R^{3N}} dy\left[ \form{\psi}{H^{(N)}}{\psi}\|\phi_y\psi\|^2+\frac{3N\pi^2}{L^2}\|\phi_y\psi\|^2\right].\label{intequation}
\end{align}
By \eqref{intequation}
$$
\int_{\R^{3N}} dy\left[\form{\phi_y\psi}{H^{(N)}}{\phi_y\psi} - \left( \frac{3N\pi^2}{L^2}+ \form{\psi}{H^{(N)}}{\psi}\right)\|\phi_y\psi\|^2\right] =0.
$$
Hence there exists $y=(y_1,\ldots , y_N)$, $y_k\in\R^3$, $1\le k\le N$, such that 
$$
\form{\phi_y\psi}{H^{(N)}}{\phi_y\psi} \le \left( \frac{3N\pi^2}{L^2}+ \form{\psi}{H^{(N)}}{\psi}\right)\|\phi_y\psi\|^2
$$
and $\|\phi_y\psi\| \ne 0$. The support of $\tilde\psi_0 :=\phi_y\psi\|\phi_y\psi\|^{-1}$ is contained in the cartesian product of $N$ boxes of sidelength $L$ and centers $y_k\in\R^3$, $1\le k\le N$, and since $L=2R/\sqrt 3$ then the support is also located in $\bigtimes_{k=1}^N B_R(y_k)$.\\

\noindent\emph{Proof of Step~2.} Proof by induction in $N$. For $N=1$ the statement is trivial. Now let us assume that for some $N$ there exist $m\le N$, balls $B_1,\ldots B_m$ and a permutation $\sigma\in S_N$ such that $\bigtimes_{k=1}^{N}B_R(y_{\sigma (k)})\subset \bigtimes_{i=1}^{m}B_i^{n_i}$ and (i), (ii) hold. For $N+1$ balls $B_R(y_k)$ two cases can arise.

Case 1: There exists $1\le i\le m$ such that $\dist (B_i, B_R(y_{N+1}))\ge R$. Then we define $B_{m+1}:= B_R(y_{N+1})$. Thus $\bigtimes_{k=1}^{N+1}B_R(y_{\tilde\sigma (k)}) \subset \bigtimes_{i=1}^{m+1}B_i^{n_i}$ for $n_{m+1} =1$ and (i), (ii) is satisfied for all balls $B_i$ and where $\tilde\sigma \in S_{N+1}$ such that $\tilde\sigma (k) = \sigma (k)$, $1\le k\le N$ and $\tilde\sigma (N+1) = N+1$.

Case 2: There exists $i_1\in \{1,\ldots m\}$ such that $\dist (B_{i_1}, B_R(y_{N+1}))< R$. Then there is a ball $B^{(1)}\supset B_{i_1}\cup B_R(y_{N+1})$ with radius $\frac{1}{2}(3(n_{i_1}+1)-1)R$. If there is a $i_2\in\{1,\ldots m\}\setminus\{i_1\}$ with $\dist (B_{i_2}, B^{(1)})< R$, then there exists a ball $B^{(2)}\supset B^{(1)}\cup B_{i_2}$ with radius $\frac{1}{2}(3(n_{i_1}+n_{i_2}+1)-1)R$. By repeating this procedure Step~2 is proved by choosing a convenient permutation.
\end{proof}

Let $n\ge 1$ and let $\Omega\subset\R^3$ be a measurable set, then we define
$$
E_n(\Omega) = \inf_{\substack{\supp\ph\subset \Omega^{n}\\ \|\ph\| =1}}\form{\ph}{H^{(n)}}{\ph}.
$$

With the help of Lemma~\ref{lm:localization} any wave function $\psi\in\QQ$ can be localized into a collection of disjoint balls. The proposition below specifies a concrete estimate for the inter-ball interactions.

\begin{prop}\label{prop:split}
Let $N$ be any positive integer and let $A,V$ satisfy (AV1). Suppose $\psi \in \QQ$ is normalized with $\supp \psi\subset\bigtimes_{i=1}^m B_i^{n_i}$. Let $B_i$ be balls with radius $R_i$ and define $d_i:=\min_{j\ne i}\dist (B_i, B_j)>0$ for $1\le i\le m$. Then
\begin{align}
\form{\psi}{H^{(N)}}{\psi} \ge \sum_{i=1}^m E_{n_i}(B_i)  + (U-2\alpha)\sum_{i<j}\sum_{\substack{s_i\in C_i\\ {\ell}_j\in C_j}}\sprod{\psi}{ \frac{1}{|x_{s_i}-x_{{\ell}_j}|}\psi}-\frac{8\alpha N}{\pi^2}\sum_{i=1}^m\left(\frac{n_i}{d_i}\right).\label{prop2}
\end{align}
$C_i$ denotes the index set of the electrons supported in $B_i$.
\end{prop}

Our proof of Proposition~\ref{prop:split} is a generalization of Lemma~3 of \cite{FLST2011a}, where the two-particle case was studied. The proof in \cite{FLST2011a} shows, that it is useful to localize the phonon field about the respective particles. The phonon field is divided into two half-spaces each including one particle. In our case we have $N$ polarons that are localized in $m$ balls $B_i$ containing $n_i$ particles. It turns out that it is suitable to split up the phonon field in such a way, that every point of it is allocated to the nearest ball. We define
$$
S_i:= \{ y\in\R^3| \dist(B_i,y)<\dist (B_j,y), j\ne i \}.
$$
Since $\dist (B_i, B_j) >0$ for $i\ne j$, the definition especially ensures $B_i\subset S_i$ for $1\le i\le m$ and
\begin{align}
\overline{\bigcup_i S_i}=\R^3,\quad\text{where}\quad S_i\cap S_j =\emptyset \text{ for } i\ne j. \label{sumsi}
\end{align}
Magnetic and electric fields that satisfy (AV1) can easily be added in Lemma~3 of \cite{FLST2011a}, since no special properties of the laplacian are needed.

\begin{proof}[Proof of Proposition~\ref{prop:split}.]
It is useful to rearrange the Fr\"ohlich hamiltonian \eqref{frhamiltonian} such that it allows for the partition into balls
\begin{align}
H^{(N)} = \sum_{i=1}^m \left(\sum_{{\ell}_i\in C_i}\left(T_{{\ell}_i}-\sqrt\alpha\phi(x_{{\ell}_i})\right)+U\sum_{\substack{s_i,{\ell}_i\in C_i\\ s_i < {\ell}_i}}\frac{1}{|x_{s_i}-x_{{\ell}_i}|}\right)  +H_{ph}+U\sum_{i<j}\sum_{\substack{s_i\in C_i\\ {\ell}_j\in C_j}}\frac{1}{|x_{s_i}-x_{{\ell}_j}|},\label{hamiltonrearr}
\end{align}
with $T_{{\ell}_i} = D_{A,x_{{\ell}_i}}^2+V(x_{{\ell}_i})$.
Define
$$
\hat a (x) = \frac{1}{(2\pi )^{3/2}}\int dk e^{ikx} a(k),\quad \hat a^{\ast} (x) = \frac{1}{(2\pi )^{3/2}}\int dk e^{-ikx} a^{\ast}(k).
$$
$\hat a (x) $ is a properly defined operator on the Fock space, but $\hat a^{\ast} (x)$ is not. Below, the operators are interpreted as quadratic forms, in which case they are well-defined. By Plancherel
\begin{align}
\phi(x)=\frac{1}{\pi^{3/2}}\int dy\frac{\hat a (y)+\hat a^{\ast} (y)}{|x-y|^2},\quad H_{ph}= \int \hat a^{\ast} (y)\hat a (y)dy\label{ft:operators}.
\end{align}
Let $1\le i\le m$, then for the surrounding $S_i$ of $B_i$ we associate the annihilation operator $\hat a_i (y)$. It is defined by
\begin{align}
\hat a_i (y)= \hat a (y)-g_i(y),\quad y\in S_i,\label{agi}
\end{align}
where
\begin{align}
g_i(y) =\frac{\sqrt{\alpha}}{\pi^{3/2}}\sum_{\substack{j=1\\ j\ne i}}^m\sum_{{\ell}_j\in C_j}\frac{1}{|x_{{\ell}_j}-y|^2}\chi_{S_i} (y).\label{gi}
\end{align}
Using \eqref{sumsi}, \eqref{ft:operators} and \eqref{agi} and the phonon energy $H_{ph}$ becomes
\begin{align}
H_{ph}= \sum_{i=1}^m\int_{S_i} dy\hat a_i^{\ast} (y)\hat a_i (y) + \sum_{i=1}^m\int_{S_i} dy (\hat a_i (y)+\hat a_i^{\ast} (y))g_i(y) + F_1,\label{hphonon}
\end{align}
with multiplication operator
\begin{align}
F_1 = \sum_{i=1}^m\|g_i\|^2.\label{errortermf1}
\end{align}
Using \eqref{sumsi} and \eqref{ft:operators} the interaction-term $\phi(x)$ splits up into two parts
\begin{align}
\sqrt\alpha \sum_{i=1}^m\sum_{{\ell}_i\in C_i}\phi(x_{{\ell}_i}) =& \frac{\sqrt{\alpha}}{\pi^{3/2}}\sum_{i=1}^m\sum_{{\ell}_i\in C_i}\int_{S_i}dy\frac{\hat a (y)+\hat a^{\ast} (y)}{|x_{{\ell}_i}-y|^2} + \sum_{i=1}^m\int_{S_i}dy(\hat a (y)+\hat a^{\ast} (y))g_i(y)\nonumber\\
=& \frac{\sqrt{\alpha}}{\pi^{3/2}}\sum_{i=1}^m\sum_{{\ell}_i\in C_i}\int_{S_i}dy\frac{\hat a_i (y)+\hat a_i^{\ast} (y)}{|x_{{\ell}_i}-y|^2} + \sum_{i=1}^m\int_{S_i}dy(\hat a_i (y)+\hat a_i^{\ast} (y))g_i(y)\nonumber
\\ &+2F_1+F_2,\label{interactionsplit}
\end{align}
where in the second step \eqref{agi} was used and
\begin{align}
F_2(x_1,\ldots x_N) = \frac{2\sqrt{\alpha}}{\pi^{3/2}}\sum_{i=1}^m\sum_{s_i\in C_i}\int_{S_i} dy\frac{g_i(y)}{|x_{s_i}-y|^2}.\label{errortermf2}
\end{align}
Inserting \eqref{hphonon} and \eqref{interactionsplit} in \eqref{hamiltonrearr}, we obtain that
\begin{align}
H^{(N)} = \sum_{i=1}^mK_i + U\sum_{i<j}\sum_{\substack{s_i\in C_i\\ {\ell}_j\in C_j}}\frac{1}{|x_{s_i}-x_{{\ell}_j}|}-(F_1+F_2),\label{hn-ki}
\end{align}
where
$$
K_i = \sum_{{\ell}_i\in C_i}\left(T_{{\ell}_i}-\frac{\sqrt{\alpha}}{\pi^{3/2}}\int_{S_i}dy\frac{\hat a_i(y)+\hat a_i^{\ast} (y)}{|x_{{\ell}_i}-y|^2}\right) +\int_{S_i}dy\hat a_i^{\ast} (y)\hat a_i (y)+U\sum_{\substack{s_i,{\ell}_i\in C_i\\ s_i < {\ell}_i}}\frac{1}{|x_{s_i}-x_{{\ell}_i}|}.
$$
Let $\psi \in \QQ $ be normalized and $\supp \psi\subset \bigtimes_{i=1}^m B_i^{n_i}$, then by \eqref{hn-ki}
\begin{align}
\form{\psi}{H^{(N)}}{\psi} &= \sum_{i=1}^m\form{\psi}{K_i}{\psi} + U\sum_{i<j}\sum_{\substack{s_i\in C_i\\ {\ell}_j\in C_j}}\sprod{\psi}{\frac{1}{|x_{s_i}-x_{{\ell}_j}|}\psi}-\sprod{\psi}{(F_1+F_2)\psi}.\label{proof:lastequ}
\end{align}
A bound for $\sprod{\psi}{(F_1+F_2)\psi}$ is derived in Lemma~\ref{lm:errorterms}. It remains to bound $\form{\psi}{K_i}{\psi}$. Because $L^2(\R^3) = \bigoplus_{i=1}^m L^2(S_i)$ the corresponding symmetric Fock space satisfies $\FF = \bigotimes_{i=1}^m \FF_{S_i}$, where $\FF_{S_i} := \FF (L^2(S_i))$. Since the precise form of $K_i$ is $1\otimes\ldots\otimes K_i\otimes\ldots\otimes 1$ on $\bigotimes_{i=1}^m L^2(\R^{3n_i})\otimes \FF_{S_i} = L^2(\R^{3N})\otimes \FF$ and since $\bigotimes_{i=1}^m \FF_{S_i,0}\otimes C_0^{\infty} (\R^{3n_i})$ is a form-core, the proof of the theorem is a consequence of Lemma~\ref{Kiabsch}.
\end{proof}

The key ingredients Lemma~\ref{lm:localization} and Proposition~\ref{prop:split}, together with the generalization of \cite{LT1997} from Section~\ref{sec:compsup} enables us to proof Theorem~\ref{mainthm}.

\begin{proof}[\textbf{Proof of Theorem~\ref{mainthm}.}]
Let $\psi\in\QQ$ be normalized. By Lemma~\ref{lm:localization} and Proposition~\ref{prop:split} there exists a constant $ C>0$
\begin{align}
\form{\psi}{H^{(N)}}{\psi}\ge \sum_{i=1}^m E_{n_i}(B_i) -C\alpha\frac{N^2}{R}-\frac{9N\pi^2}{4R^2}.\label{localandintercluster}
\end{align}
We choose $R=N^{-1}\alpha^{-19/23}$ and since we use scaled fields $A_{\alpha},V_{\alpha}$, by Corollary~\ref{cor:energyball} there exists a constant $c(A,V)$
\begin{align}
E_{n_i}(B_i)\ge \alpha^2 C_{n_i}(A,V,\nu,1) - c(A,V)\alpha^{42/23}n_i^3\left(1+\frac{n_i^2}{N^2}\right).\label{enibi:bound}
\end{align}
Statement a) of the theorem is then a consequence of \eqref{localandintercluster}, \eqref{enibi:bound}, of the fact $\sum_{i=1}^m n_i^q\le N^q$ for $q\ge 1$ and the energy inequality \eqref{generalenergyin}.
If the fields $A,V$ are not rescaled, an analogous calculation as above proves \eqref{eqn:thmb}, where Corollary~\ref{cor:energyball} and Lemma~\ref{limnotscaled} are used. 
\end{proof}

\begin{lemma}\label{Kiabsch}
Suppose $1\le i\le m$. Let $\psi_i\in L^2( B_i^{n_i})\otimes \FF_{S_i}$ be normalized. Then
$$
\form{\psi_i}{K_i}{\psi_i}\ge E_{n_i}(B_i).
$$
\end{lemma}

\begin{proof}
Let $\Omega_i$ be the normalized vacuum of $\FF_{S_i^c}$. 
Let $g_i$ be defined by \eqref{gi} and  $\hat a (g_i)$ by \eqref{agi}, then $W(g_i)=e^{\hat a^{\ast} (g_i) - \hat a (g_i)}$ is a unitary operator acting on $\FF$. Further it satisfies $W(g_i)\hat a(y) =\hat a_i(y)W(g_i)$, and therefore in the sense of quadratic forms
\begin{align}
W(g_i)H^{(n_i)}W(g_i)^{-1} = K_i + \int_{S_i^c}dy\left[\hat a^{\ast}(y) \hat a(y) - \frac{\sqrt{\alpha}}{\pi^{3/2}}\sum_{{\ell}_i\in C_i}\frac{\hat a(y)+\hat a^{\ast}(y)}{|x_{{\ell}_i}-y|^2}\right],\label{unitaryhni}
\end{align}
which follows by \eqref{hamiltonrearr}, \eqref{ft:operators}. By \eqref{unitaryhni}
\begin{align*}
E_{n_i}(B_i) &\le \form{\psi_i\otimes \Omega_i}{W(g_i)H^{(n_i)}W(g_i)^{-1}}{\psi_i\otimes \Omega_i}=\form{\psi_i}{K_i}{\psi_i}. 
\end{align*}
\end{proof}

\begin{lemma}\label{lm:errorterms}
Suppose $\psi$ satisfies the assumptions of Proposition~\ref{prop:split}. Then for $F_1$ and $F_2$ defined in \eqref{errortermf1} and \eqref{errortermf2}
\begin{enumerate}
\item[a)] $\sprod{\psi}{F_1\psi} \le N\frac{8\alpha}{\pi^2}\sum_{i=1}^m\frac{n_i}{d_i}\|\psi\|^2$.
\item[b)] $\sprod{\psi}{F_2\psi}  \le 2\alpha \sum_{i<j}\sum_{\substack{s_i\in C_i\\ {\ell}_j\in C_j}}\sprod{\psi}{\frac{1}{|x_{s_i}-x_{{\ell}_j}|}\psi}$.
\end{enumerate}
\end{lemma}

\begin{proof}
a) By Cauchy-Schwarz
\begin{align}
\left(\sum_{\substack{j=1\\ j\ne i}}^m\sum_{{\ell}_j\in C_j}\frac{1}{|x_{{\ell}_j}-y|^2}\right)^2 
\le N \sum_{\substack{j=1\\ j\ne i}}^m\sum_{{\ell}_j\in C_j}\frac{1}{|x_{{\ell}_j}-y|^4}.\label{f1:vorb}
\end{align}
By the definition of $F_1$ and \eqref{f1:vorb}
\begin{align}
F_1(x_1,\ldots x_N) &\le \frac{\alpha}{\pi^3}N\sum_{i=1}^m \sum_{\substack{j=1\\ j\ne i}}^m\sum_{{\ell}_j\in C_j}\int_{S_i}\frac{1}{|x_{{\ell}_j}-y|^4}dy\nonumber\\
&=\frac{\alpha}{\pi^3}N\sum_{j=1}^m \sum_{{\ell}_j\in C_j}\int_{S_j^c}\frac{1}{|x_{{\ell}_j}-y|^4}dy.\label{f1:ineq}
\end{align}
In the last step we exchanged the sums with respect to $i$ and $j$ and used that $\sum_{\substack{i=1\\ i\ne j}}^m \chi_{S_i} = \chi_{S_j^c}$. Since $x_{{\ell}_j} \in B_j$
\begin{align}
\int_{S_j^c}\frac{1}{|x_{{\ell}_j}-y|^4}dy \le \int_{B_{d_j/2}^c(0)}\frac{1}{|y|^4}dy = \frac{8\pi}{d_j}.\label{ineq:generalfact}
\end{align}
\eqref{ineq:generalfact} and \eqref{f1:ineq} now conclude a).\\

\noindent b) By the definition of $F_2$ and $g_i$
\begin{align*}
F_2(x_1,\ldots x_N) &= \frac{2\alpha}{\pi^3} \sum_{i<j}\int dy\left(\sum_{{\ell}_j\in C_j}\frac{1}{|x_{{\ell}_j}-y|^2}\right)\left(\sum_{s_i\in C_i}\frac{1}{|x_{s_i}-y|^2}\right) \left( \chi_{S_i}(y)+\chi_{S_j}(y)\right)\\
&\le \frac{2\alpha}{\pi^3} \sum_{i<j}\sum_{\substack{s_i\in C_i\\ {\ell}_j\in C_j}}\int dy\frac{1}{|x_{{\ell}_j}-y|^2}\frac{1}{|x_{s_i}-y|^2} \\
&= 2\alpha\sum_{i<j}\sum_{\substack{s_i\in C_i\\ {\ell}_j\in C_j}}\frac{1}{|x_{s_i}-x_{{\ell}_j}|}.
\end{align*}
In the second step $\chi_{S_i}(y)+\chi_{S_j}(y) \le 1$ for $i\ne j$ was used. The last equality follows by direct integration using cylindrical coordinates.
\end{proof}

\section{Compactly Supported Multipolarons}\label{sec:compsup}

The objective of this section is to bound the energy of compactly supported $N$-polarons in electromagnetic fields by the respective $N$-particle Pekar-Tomasevich functional from below. The Pekar-Tomasevich functional with external electric and magnetic fields for $N$ particles acting on $L^2(\R^{3N})$ is defined by
\begin{align}
\EE^{(N)}_{U,\alpha}(A,V,\ph) = \int_{\R^{3N}}dx\sum_{j=1}^N \left(|D_{A,x_j}\ph|^2+V(x_j)|\ph|^2\right)+U V_C(x_1,\ldots x_N)|\ph|^2 - \alpha D(\rho_{\ph}),\label{PT-functional}
\end{align}
with density
$$
\rho_{\ph}(x) = \sum_{j=1}^N\int_{\R^3} dx|\ph (x_1,\ldots, x_{j-1},x,x_{j+1},\ldots x_N)|^2 dx_1\ldots \widehat{dx_j}\ldots dx_N,
$$
and self-interaction term
\begin{align}
D(\rho) = \int_{\R^6}\frac{\rho(x)\rho(y)}{|x-y|}dxdy.
\end{align}
\eqref{PT-functional} coincides with the definition \eqref{PT-derivation}. This fact can be seen by choosing $\eta$ to be the coherent state that is determined by $a(k)\eta = -\overline{f(k)}\eta$ and $f(k) = 2\sqrt{\alpha\pi}\hat{\rho}_{\ph}(k)|k|^{-1}$, then $\eta$ minimizes $\xi\mapsto\form{\ph\otimes\xi}{H^{(N)}}{\ph\otimes\xi}$. The argument goes back to Pekar \cite{P1954}. 

This section is in principle based on \cite{LT1997}. For completeness reasons the main ideas of the proofs are performed nevertheless. Since we allow for general $A,V$ the translation invariance of the Fr\"ohlich, and the Pekar and Tomasevich model is abolished. The translation invariance seems to play some role in \cite{LT1997}, however after a little modification it is not necessary for our proof to work. In the case $N=1$ this issue was already noticed in \cite{GW2013}. 

Let $h\in L^2(\R^3)$, then $a(h):= \int dk h(k)a(k)$ is a well-defined Fock space operator. Suppose $\Omega\subset\R^3$ be a measurable set, then
$$
N_{\Omega} = \int_{\Omega}dk a^{\ast}(k) a(k).
$$
Let $h\in L^2(\R^3)$ be normalized and $\supp h\subset \Omega$, then in the sense of quadratic forms
\begin{align}
a^{\ast}(h)a(h)\le N_{\Omega}.\label{ineq:Fockspop}
\end{align}
Let $B_{\Lambda}:= B_{\Lambda}(0)$ denote the ball centered in the origin with arbitrary radius $\Lambda>0$. Suppose $\beta = 1-\frac{8\alpha N}{\pi\Lambda}$. For any positive integer $N$ we define
\begin{align}
H_{\Lambda}^{(N)} := \sum_{j=1}^N\left(\beta D_{A,x_j}^2+V(x_j) +\sqrt{\alpha}(a(f_{x_j}) + a^{\ast}(f_{x_j}))\right)+\beta UV_C(x_1,\ldots x_N)+N_{B_{\Lambda}},
\end{align}
where $f_x(k)=\chi_{B_{\Lambda}}(k)|k|^{-1}e^{-ikx}$.

\begin{lemma}\label{lm:cutoff}
For any values of $N$ and $\Lambda >0$ in the sense of quadratic forms on $\QQ$
\begin{align}
H^{(N)}\ge H_{\Lambda}^{(N)} -\frac{1}{2}.\label{ineq:cutoff}
\end{align}
$H^{(N)}$ can be interpreted as a selfadjoint Operator.
\end{lemma}

\begin{proof}
Let $\psi\in\QQ$. The interaction term of the Fr\"ohlich hamilton operator is rewritten in terms of Fock space operators
\begin{align}
\form{\psi}{\phi(x)}{\psi} = \sprod{\psi}{(a(f_x)+\sum_{l=1}^3[D_{A,l},a(g_{l,x})])\psi} +c.c.,\label{eqn:interactionterm}
\end{align}
where $D_{A,l}$ denotes the $l$-th component of $D_A$ and
$$
g_{l,x}(k)=\frac{1}{\sqrt 2\pi}\chi_{B_{\Lambda}^c} (k) \frac{e^{ikx} k_l}{|k|^3}.
$$
For every $\eps_1, \eps_2>0$
\begin{align}
\sqrt{\alpha}|\sprod{\psi}{a(f_x)\psi}| &\le \frac{\eps_1}{2}\sprod{\psi}{N_{B_{\Lambda}}\psi} +\frac{\Lambda\alpha}{\pi\eps_1}\|\psi\|^2,\label{af-bound}\\
\sqrt{\alpha}\left|\sprod{\psi}{\sum_{l=1}^3[D_{A,l},a(g_{l,x})]\psi}\right| &\le  \frac{\eps_2}{2}\sprod{\psi}{D_A^2\psi}+\frac{4\alpha}{\eps_2\pi\Lambda}\sprod{\psi}{N_{B_{\Lambda}^c}\psi} + \frac{2\alpha}{\eps_2\pi\Lambda}\|\psi\|^2\label{ag-bound}.
\end{align}
\eqref{af-bound} is a consequence of the Cauchy-Schwarz inequality and \eqref{ineq:Fockspop}. For the proof of \eqref{ag-bound} see \cite{LT1997}.

Suppose $\eps_2 = 8\alpha N/ (\Lambda\pi )$, then \eqref{eqn:interactionterm} and \eqref{ag-bound} imply \eqref{ineq:cutoff}.

Let $\eps_1= N^{-1}\eps_2$ and $\Lambda = 8\alpha N/(\eps_2^2\pi)$, then by \eqref{af-bound} and \eqref{ag-bound} $\sqrt\alpha\sum_{j=1}^N\form{\psi}{\phi(x_j)}{\psi}$ is relatively form-bounded by $\sum_{i=1}^N D_{A,x_i}^2 +N$ with bound $\eps_2$. Hence there exists a corresponding self-adjoint operator with form core $\QQ$ (see \cite{RS2}), where in addition \eqref{Vformbound} is used.
\end{proof}

\begin{prop}\label{prop:liebthomas}
Suppose $N>0$ is an integer. Let $\psi\in \QQ$ be normalized and $\supp\psi\subset B_r(y)^N$ for any $r>0$ and $y\in\R^3$. Then for arbitrary $P, \Lambda>0$
\begin{align}
\form{\psi}{H^{(N)}}{\psi} \ge \beta C_N(A,\beta^{-1}V, U,\alpha\beta^{-2})-\frac{6N^2\alpha P^2 r^2\Lambda}{(1-\beta)\pi}-\frac{1}{2}-\left(2\frac{\Lambda}{P} +1\right)^3.\label{energyestimateball}
\end{align}
\end{prop}
\begin{proof}
Since $C_N$ is constant with respect to translations of the potentials, i.e. $A(.-y)$ and $V(.-y)$, we may assume that $y=0$. By Lemma~\ref{lm:cutoff} the left-hand side of \eqref{energyestimateball} is estimated from below by $\sprod{\psi}{H^{(N)}_{\Lambda}\psi}$ with error $\frac{1}{2}$.

In the next step the modes are replaced by the so called block modes, of which only finitely many exist. For a given $P>0$, we define
\begin{align*}
B(n) &:=  \{k\in B_{\Lambda}|k_i-n_iP|\le P/2\},\quad n\in \Z^3,\\
\Lambda_P &:= \{n\in\Z^3|B(n)\ne \emptyset\}.
\end{align*}
In every $B(n)$ an arbitrary $k_n$ is chosen, they are specified later. The block modes are defined by
$$
a_n:=\frac{1}{M_n}\int_{B(n)}\frac{dk}{|k|}a(k),\quad M_n=\left(\int_{B(n)}\frac{dk}{|k|^2}\right)^{1/2}.
$$
They are well-defined normalized annihilation operators acting on the Fock space $\FF$. For random $\delta >0$
\begin{align}
H_{block}^{(N)} =& \sum_{j=1}^N\left(\beta D_{A,x_j}^2+V(x_j)+\frac{\sqrt{\alpha}}{\sqrt 2\pi}\sum_{n\in\Lambda_P}M_n\left(e^{ik_nx_j}a_n+e^{-ik_nx_j}a_n^{\ast}\right)\right)\nonumber\\
+&\beta U V_C(x_1,\ldots x_N)+(1-\delta) N_{block},
\end{align}
where $N_{block} = \sum_{n\in\Lambda_P} a_n^{\ast}a_n$. 

Next we show
\begin{align}
\sprod{\psi}{H^{(N)}_{\Lambda}\psi}\ge\inf_{\substack{\tilde\psi\in\QQ\\ \|\tilde\psi\|=1}}\sup_{\{k_n\}}\sprod{\tilde\psi}{H_{block}^{(N)}\tilde\psi}-\frac{6N^2\alpha P^2 r^2\Lambda}{\delta\pi}  .\label{ineq:block}
\end{align}
Let $1\le j\le N$, then in the sense of quadratic forms 
\begin{align}
&\frac{\delta}{N} N_{B(n)} +\frac{\sqrt \alpha}{\sqrt 2 \pi}\int_{B(n)} \frac{dk}{|k|}\left( (e^{ikx_j}-e^{ik_nx_j})a(k)+(e^{-ikx_j}-e^{-ik_nx_j})a^{\ast}(k)\right) \nonumber\\
\ge &-\frac{N\alpha}{2\pi^2\delta}\int_{B(n)}dk\frac{|e^{ikx_j}-e^{ik_nx_j}|^2}{|k|^2},\label{ineq:knquaderg}
\end{align}
which follows by completion of squares. Let $k\in B(n)$ and $|x_j|<r$, $1\le j\le N$
\begin{align}
|e^{ikx_j}-e^{ik_nx_j}|^2\le 3P^2r^2,\quad 1\le j\le N\label{ineq:eikneik}.
\end{align}
\eqref{ineq:block} is a consequence of \eqref{ineq:knquaderg} summed over all $n\in \Lambda_P$, $\sum_{n\in\Lambda_P} a_n^{\ast}a_n \le N_{B_{\Lambda}}$ and \eqref{ineq:eikneik}. 

It remains to prove that for all normalized $\psi\in\QQ$
\begin{align}
\sup_{\{k_n\}}\sprod{\psi}{H_{block}^{(N)}\psi}\ge \beta C_N(A,\beta^{-1}V, U, \alpha\beta^{-2}) - |\Lambda_P|.
\end{align}
To do so, the block operators $a_n$ are replaced by complex numbers $z_n$ using coherent states. The closed subspace $M:=\text{span}\{\chi_{B(n)}|.|^{-1}|n\in \Lambda_P\}\subset L^2(\R^3)$ generates the symmetric Fock space $\FF (M)$, i.e. the Fock space that is constructed by the block operators $a_n^{\ast}$, $n\in\Lambda_P$. Since $M$ is a closed subspace
\begin{align*}
\FF=\FF(M\oplus M^{\perp})\cong \FF(M)\otimes\FF(M^{\perp}).
\end{align*}
Suppose $z=(z_n)_{n\in \Lambda_P}$, $z_n\in\C$, then we define normalized coherent states $\eta_z\in \FF (M)$ by
\begin{align}
\eta_z:=\prod_{n\in\Lambda_P}e^{z_na_n^{\ast}-\overline{z_n}a_n}\Omega,\label{def:cohrentstate}
\end{align}
where $\Omega\in\FF (M)$ denotes the normalized vacuum. From \eqref{def:cohrentstate} $a_n\eta_z =z_n\eta_z$. If $\psi\in \QQ$ normalized and $\psi_z = \sprod{\eta_z}{\psi}$, note that the inner product acts on $\FF (M)$, then $\psi_z\in L^2(\R^{3N})\otimes\FF(M^{\perp})$. For notational simplicity hereinafter the inner products are not labeled explicitly. By a short calculation in the sense of weak integrals on the Fock space $\FF (M)$ for $dz= \prod_{n\in\Lambda_P}\frac{1}{\pi}\int dx_ndy_n$
\begin{align}
&\int dz\sprod{.}{\eta_z}\eta_z =1, &&\int dz z_n\sprod{.}{\eta_z}\eta_z =a_n,\nonumber\\
&\int dz(|z_n|^2-1)\sprod{.}{\eta_z}\eta_z =a_n^{\ast}a_n, && \int dz \overline{z_n}\sprod{.}{\eta_z}\eta_z =a_n^{\ast},\label{identities}
\end{align}
where the last equality follows from the first one and the fact $[a_n,a_n^{\ast}]=1$ for all $n\in\Lambda_P$. Let the block modes be replaced by the identities \eqref{identities}, then
\begin{align}
\sprod{\psi}{H_{block}^{(N)}\psi} = \int dz\sprod{\psi_z}{h_z\otimes 1\psi_z},\label{int:hz}
\end{align}
whereas $h_z$ is a Schr\"odinger operator on $L^2(\R^{3N})$
\begin{align*}
h_z &=\sum_{j=1}^N (\beta D_{A,x_j}^2+V(x_j)) + \beta U V_C(x_1,\ldots x_N)+ (1-\delta)\sum_{n\in\Lambda_P}(|z_n|^2-1)\\
&+\frac{\sqrt{\alpha}}{\sqrt 2\pi}\sum_{j=1}^N\sum_{n\in\Lambda_P}M_n\left(z_n e^{ik_nx_j}+\overline{z_n}e^{-ik_nx_j} \right).
\end{align*}
Since $\rho_z(x) :=\sum_{j=1}^N\int_{\R^{3(N-1)}}|\psi_z(x_1,\ldots,x_{j-1},x,x_{j+1},\ldots ,x_N)|^2dx_1\ldots \widehat{dx_j}\ldots dx_N$, then
$$
\sum_{j=1}^N \sprod{\psi_z}{e^{-ikx_j}\psi_z} = (2\pi)^{3/2}\hat{\rho}_z (k).
$$
Obviously
\begin{align}
\inf_{\{k_n\}}\int dz |\hat{\rho}_z(k_n)|^2\|\psi_z\|^2 \le \int dz |\hat{\rho}_z(k)|^2\|\psi_z\|^2,\quad\forall k\in B(n).\label{ineq:knk}
\end{align}
By completion of squares of \eqref{int:hz} with respect to $z_n$ and \eqref{ineq:knk}
\begin{align*}
&\sup_{\{k_n\}}\int dz\sprod{\psi_z}{h_z\psi_z}\\
\ge &\int dz\sprod{\psi_z}{\left[\sum_{j=1}^N(\beta D_{A,x_j}^2+ V(x_j))+ \beta U V_C(x_1,\ldots x_N)\right]\psi_z}\\
&-\frac{4\pi\alpha}{(1-\delta)}\int dz\int_{B_{\Lambda}}dk\frac{|\hat{\rho}_z (k)|^2}{\|\psi_z\|^2|k|^2}-|\Lambda_P|\\
\ge &\int dz\sprod{\psi_z}{\left[\sum_{j=1}^N(\beta D_{A,x_j}^2+ V(x_j)) +\beta U V_C(x_1,\ldots x_N)\right]\psi_z}\\
&-\frac{\alpha}{(1-\delta)}\int dz\int\frac{\rho_z(x)\rho_z(y)}{\|\psi_z\|^2|x-y|}dxdy-|\Lambda_P|.
\end{align*}

The integrand is estimated from below by
$$
\beta \|\psi_z\|^2C_N(A,\beta^{-1}V, U,\alpha\beta^{-2}),
$$
where $\delta = 1-\beta$ has been chosen. The assumption follows by $\int\|\psi_z\|^2dz =1$ and $|\Lambda_P|\le\left(2\frac{\Lambda}{P}+1\right)^3$.
\end{proof}

Next we evaluate \eqref{energyestimateball} on a single localized $n$-polaron found in Lemma~\ref{lm:localization}. The constants $\Lambda$ and $P$ can be chosen freely, but the radius $r$ of the corresponding ball is determined by Lemma~\ref{lm:localization}~(ii), i.e. $r=\frac{1}{2} (3n-1)R$ for any $R>0$ fixed.

\begin{corollary}\label{cor:energyball}
Let $\nu>0$ be arbitrary and let $n>0$ be any integer. Let $R>0$ and let $B$ be a ball of radius $\frac{1}{2} (3n-1)R$. Suppose $A,V$ satisfy (AV1) and \eqref{generalenergyin}, and let them be scaled by $A_{\alpha}(x) = \alpha A(\alpha x)$, $V_{\alpha}(x) = \alpha^2 V(\alpha x)$. Then there exists $c(A,V)$
\begin{align}
E_n(B)\ge \alpha^2 C_n(A,V,\nu,1)- 3 R^2 \alpha^{80/23}n^5 -c(A,V)\alpha^{42/23}n^3. \label{kor:enestimatecn}
\end{align}
Moreover, if $A,V$ satisfy (AV2), then $c(A_{\alpha^{-1}},V_{\alpha^{-1}})$ is uniformly bounded for $\alpha$ large.
\end{corollary}

\begin{proof}
Since $\lambda\mapsto C_n(A,\lambda V, \nu, \lambda^2)$ is concave, the one-sided derivatives exist and
\begin{align}
C_n(A,\beta^{-1}V,\nu,\beta^{-2}) \ge C_n(A,V,\nu,1) + (\beta^{-1}-1)\frac{d}{d\lambda}C_n(A,\lambda V,\nu,\lambda^2)\Big|_{\lambda =2-} .\label{concavedev}
\end{align}
The derivation term of \eqref{concavedev} is estimated by Lemma~\ref{lm:ndependence}. The statement is then a consequence of Proposition~\ref{prop:liebthomas}, the scaling property \eqref{scalingproperty}, \eqref{concavedev} and Lemma~\ref{lm:ndependence}, where we determine the free parameters $\Lambda =n\alpha^{27/23}$, $P= \alpha^{13/23}$ and hence $\left(2\frac{\Lambda}{P}+1\right)^3\le 9 n^3\alpha^{42/23}$ and $1-\beta =\frac{8}{\pi} \alpha^{-4/23}$.

If the fields $A,V$ are not rescaled, then in \eqref{concavedev} $A,V$ are replaced by $A_{\alpha^{-1}}, V_{\alpha^{-1}}$. Hence the second statement follows by Lemma~\ref{lm:ndependence}.
\end{proof}

\section{Appendix}

\begin{lemma}\label{lm:ndependence}
For any values of $N$ and $\nu >0$. Suppose (AV1) and \eqref{generalenergyin} are satisfied, then there exists $c(A,V)$
\begin{align}
 c (A,V) N\ge C_N(A,V,\nu,1) \ge - c(A,V) N^3,\quad N\in\N.\label{ptminimalesti}
\end{align}
Moreover, if (AV2) is satisfied, then there exists a constant $c>0$ such that $c(A_{\alpha^{-1}},V_ {\alpha^{-1}})\le c$ for $\alpha$ large enough.
\end{lemma}
\noindent\emph{Remark.} In the physical regime $\nu >2$ without external fields, in \cite{FLST2011a} it was proven that $C_N(0,0,\nu,1)\ge -c(\nu)N$ for $c(\nu)>0$.
\begin{proof}[Proof of Lemma~\ref{lm:ndependence}.]
Let $\nu>0$, then by \eqref{generalenergyin}
\begin{align}
C_N(A,V,\nu,1)\le N C_1(A,V,1),\label{energysubadd}
\end{align}
which proves the upper bound in \eqref{ptminimalesti}. Let $\ph\in C_0^{\infty}(\R^{3N})$ be normalized, then by the Hardy and the diamagnetic inequality
\begin{align}
D(\rho_{\ph})\le 2N^{3/2}\left(\sum_{j=1}^N\|D_{A,x_j}\ph\|^2\right)^{1/2}. \label{Dhardy}
\end{align}
Thus from \eqref{Dhardy}, \eqref{Vformbound} and by completion of squares with respect to $\left(\sum_{j=1}^N\|D_{A,x_j}\ph\|^2\right)^{1/2}$, we conclude
$$
\EE^{(N)}_{\nu,1}(A,V,\ph) \ge -\frac{N^3}{(1-\eps)}-C_{\eps} N,
$$
where $C_{\eps}>0$. This proves the lower bound of \eqref{ptminimalesti}. 

By a similar calculation
\begin{align}
\EE^{(N)}_{\nu,1}(A_{\alpha^{-1}},V_{\alpha^{-1}},\ph) \ge -\frac{N^3}{(1-\eps)}-C_{\eps}\alpha^{-2} N. \label{simicalc}
\end{align}
For $\alpha$ large, there exists a constant $c>0$ such that \eqref{simicalc} is bounded from below by $-cN^3$. Lemma~\ref{limnotscaled} and \eqref{energysubadd} finish the proof.
\end{proof}

\begin{lemma}\label{lm:weakfieldlimit1}
Suppose $A\in L^3_{\rm loc}(\R^3)$ and $V\in L^{3/2}_{\rm loc}(\R^3)$. Then 
\begin{align*}
A_{\alpha^{-1}} &\to 0\quad (\alpha\to \infty)\quad \text{in } L^2_{\rm loc}(\R^3),\\
V_{\alpha^{-1}} &\to 0\quad (\alpha\to \infty)\quad \text{in } L^1_{\rm loc}(\R^3).
\end{align*}
\end{lemma}
For the proof of this Lemma we refer to \cite{GW2013}. 

\begin{lemma}\label{limnotscaled}
For any values of $N$ and $\nu >0$. If the assumptions (AV2) are satisfied, then
$$
\lim_{\alpha\to\infty} C_N(A_{\alpha^{-1}}, V_{\alpha^{-1}},\nu,1) = C_N(0,0,\nu,1).
$$
\end{lemma}

\begin{proof}[Proof of Lemma~\ref{limnotscaled}]
For any normalized $\ph\in C_0^{\infty}(\R^{3N})$
\begin{align*}
\limsup_{\alpha\to\infty} C_N(A_{\alpha^{-1}}, V_{\alpha^{-1}},\nu,1)\le \limsup_{\alpha\to\infty} \EE_{\nu,1}^{(N)}(A_{\alpha^{-1}}, V_{\alpha^{-1}},\ph) = \EE_{\nu,1}^{(N)} (0,0,\ph),
\end{align*}
where the last equality is a consequence of Lemma~\ref{lm:weakfieldlimit1}. This implies 
$$
\limsup_{\alpha\to\infty} C_N(A_{\alpha^{-1}}, V_{\alpha^{-1}},\nu,1)\le C_N(0,0,\nu,1).
$$ 

It remains to prove the other direction. For all normalized $\ph\in C_0^{\infty}(\R^{3N})$ and by \eqref{Vformbound} for all $1>\eps >0$ there exists $C_{\eps}>0$ such that
\begin{align}
\EE_{\nu,1}^{(N)}(A_{\alpha^{-1}}, V_{\alpha^{-1}},\ph) &\ge (1-\eps)C_N(0,0,\nu,(1-\eps )^{-1})-C_{\eps} N \alpha^{-2}\label{lowerboundencn}
\end{align}
where the diamagnetic inequality has been used. The Lemma follows immediately from \eqref{lowerboundencn}.
\end{proof}

\bibliographystyle{plain}

\end{document}